\documentclass[runningheads]{llncs}
\usepackage[T1]{fontenc}
\usepackage{amsmath,amssymb,,amsfonts,upgreek}
\usepackage[colorlinks]{hyperref}

\newcommand*{\R}{\ensuremath{\mathbb{R}}}
\newcommand{\beq}{\begin{equation}}
\newcommand{\eeq}{\end{equation}}
\newcommand*{\cA}{{\mathcal{A}}}
\newcommand*{\de}{{\mathrm{d}}}

\begin{document}
\title{Madelung transform and variational asymptotics in Born-Oppenheimer molecular dynamics\thanks{This work was made possible through the support of Grant 62210 from the John Templeton Foundation.
	The opinions expressed in this publication are those of the authors and do not necessarily reflect the views of the John Templeton Foundation.}}

\author{Paul Bergold\inst{1}\orcidID{0000-0003-3033-0732} \and
Cesare Tronci\inst{1,2}\orcidID{0000-0002-8868-8027}}

\institute{Department of Mathematics, University of Surrey, Guildford, UK
\email{p.bergold@surrey.ac.uk}\\[3mm] \and
\mbox{Department of Physics and Engineering Physics, Tulane University, New Orleans, USA}
\email{c.tronci@surrey.ac.uk}}

\maketitle
\begin{abstract}
	While Born-Oppenheimer molecular dynamics (BOMD) has been widely studied by resorting to powerful methods in mathematical analysis, this paper presents a geometric formulation in terms of Hamilton's variational principle and Euler-Poincar\'{e} reduction by symmetry.
	Upon resorting to the Lagrangian hydrodynamic paths made available by the Madelung transform, we show how BOMD arises by applying asymptotic methods to the variational principles underlying different continuum models and their particle closure schemes.
	In particular, after focusing on the hydrodynamic form of the fully quantum dynamics, we show how the recently proposed bohmion scheme leads to an on-the-fly implementation of BOMD.
	In addition, we extend our analysis to models of mixed quantum-classical dynamics.
	
	\keywords{Born-Oppenheimer approximation \and Variational principle \and Mixed quantum-classical dynamics}
\end{abstract}
%

%%%%%%%%%%%%%%%%%%%%%%%%%%%%%%%%%%%%%%%%%%%%%%%%%%
%%%%%%%%%%%%%%%%%%%%%%%%%%%%%%%%%%%%%%%%%%%%%%%%%%
%%%%%%%%%%%%%%%%%%%%%%%%%%%%%%%%%%%%%%%%%%%%%%%%%%
\section{Introduction}
Although many-body quantum simulations have greatly benefited from high-performance computing, large molecular systems continue to pose formidable challenges.
In particular, molecular dynamics deals with systems comprising $N$ nuclei with masses $M_n>0$ and $L$ electrons with masses $m_e$.
The position coordinates of the nuclei and electrons are $r=(r_1,\dots,r_N)\in\R^{3N}$ and $x=(x_1,\dots,x_L)\in\R^{3L}$, respectively.
The time evolution of the molecular wavefunction $\Psi\in L^2(\R^{3N+3L})$ describing the system is determined by the time-dependent Schr\"odinger equation (TDSE) $i\hbar\partial_t\Psi=\widehat H_{mol}\Psi$.
In the absence of external fields, the molecular Hamiltonian reads $\widehat H_{mol}=\widehat T_e+\widehat T_n+\widehat V_{ee}+\widehat V_{en}+\widehat V_{nn}$.
The first two terms are the electronic and nuclear kinetic energy operators (i.e., $\widehat T_n=-{\hbar^2}\sum_{n=1}^NM_n^{-1}\Delta_{r_n}/2$, and similarly for $\widehat T_e$), and the remaining terms comprise the Coulomb interactions.

\paragraph{Born-Oppenheimer theory.}
Given the computational complexity of this problem, one is forced to look for additional assumptions alleviating the curse of dimensions.
One of the best-known approaches is the (time-dependent) Born-Oppenheimer (BO) theory from 1927, which stands as the cornerstone of modern quantum chemistry.
This theory is based on the observation that the large nuclear-to-electron mass ratio yields a large difference between their kinetic energies, thereby suggesting that nuclei move much slower than the electrons.
In the literature, this argument leads to the \emph{BO approximation} $\Psi(r,x,t)\approx\Omega(r,t)\phi(x;r)$, where ${\int_{\R^{3L}}\left|\phi(x;r)\right|^2\,\mathrm{d}x=1}$ so that $\phi$ is a conditional wavefunction.
Over the decades, the first and second factor have been dubbed \emph{nuclear wavefunction} and \emph{electronic wavefunction}, respectively, even though the nuclear and electronic density matrices are not actually projection operators \cite{Foskett:2019}.
Notice that here the electronic factor is time-independent, which is the essence of the BO approximation.
In particular, $\phi(r):=\phi(\_;r)$ is chosen as an eigenstate of the electronic Hamiltonian $\widehat H_e(r):=\widehat T_e+\widehat V_{ee}+\widehat V_{en}(r)+\widehat V_{nn}(r)$, that is
\begin{align}\label{eq:EEVP}
	\widehat H_e(r)\phi(r)
	=E(r)\phi(r).
\end{align}
The resulting eigenvalues for the different nuclear positions, i.e. the energies $E(r)$, are combined to obtain what is known as a \emph{potential energy surface} (PES) $E\colon\R^{3N}\to\R$.
Then, combining the BO approximation and \eqref{eq:EEVP}, the TDSE can be separated into two smaller, consecutive subproblems.
In the first step, one solves a time-independent eigenvalue problem for each fixed nuclear configuration.
In the second step, one solves the TDSE $i\hbar\partial_t\Omega=\widehat{T}_n\Omega+E\Omega$ to determine the nuclear motion.

\paragraph{Born-Oppenheimer molecular dynamics.}
Despite its well-known limitations, the Born-Oppenheimer approximation (BOA) provides a valuable tool allowing for a significant reduction in computational costs.
Yet, the high-dimensional nature of molecular systems motivates the search for further approximations.
Based on the small value of the electron-to-nuclear mass ratio, one can perform a classical limit on the nuclear evolution, whose classical dynamics then takes place on the PES made available by the electronic eigenvalue problem \eqref{eq:EEVP}.
In this picture, one is left with Newton's equations 
\beq\label{Newton-eq}
	M_n\ddot{R}_n
	=-\nabla_{\!R_n}E(R),\qquad
	n=1,\dots,N,
\eeq
for the nuclear motion.
Known under the name of \emph{Born-Oppenheimer molecular dynamics} (BOMD), the system \eqref{eq:EEVP}-\eqref{Newton-eq} is probably the most widely used model in computational quantum chemistry and represents the object of our investigation.

Over the years, the BO theory and its molecular dynamics counterpart have been established by using several powerful methods from mathematical analysis \cite{Lasser:2020}.
Yet, a geometric characterization in terms of reduction by symmetry and variational principles is still lacking.
Motivated by the success of Geometric Mechanics in a variety of modeling efforts, here we present a series of possible geometric derivations of BOMD coming from quite different perspectives.
Much hinging on the use of diffeomorphic Lagrangian paths in continuum theories, our presentation applies asymptotic analysis to Hamilton's action principle rather than to the TDSE itself.
Differing from conventional techniques, not only does our approach enable a deeper geometric understanding, but also provides new tools for the asymptotic analysis of numerical schemes beyond the BO regime.

\paragraph{Exact factorization and Madelung transform.}
Among the tools used in our discussion, the \emph{exact factorization} (XF) plays a prominent role.
This technique extends the BO approximation to consider a \makebox{time-dependent electronic factor, i.e.,}
\begin{align}\label{eq:XFA}
	\Psi(r,x,t)
	=\Omega(r,t)\phi(x,t;r),
	\quad\text{where}\quad
	\int_{\R^{3L}}\left|\phi(x,t;r)\right|^2\,\mathrm{d}x
	=1.
\end{align}
Going back to von Neumann's celebrated book, the XF has become increasingly popular in quantum chemistry.
For example, BOMD was shown to be recovered from the XF equations for $\Omega$ and $\phi$ in \cite{Eich:2016}.
Notice that \eqref{eq:XFA} provides a convenient representation of exact solutions of the TDSE.
Based on the variational setting of the TDSE, here we will combine the XF with the \emph{Madelung transform} $\Omega=\sqrt{D}e^{iS/\hbar}$, in such a way that the resulting hydrodynamic formulation allows to take advantage of the associated Lagrangian trajectories.

This paper presents four different variational approaches to the formulation of BOMD.
The first two are based on a quantum hydrodynamic description of the molecular evolution.
While one of these approaches involves the continuum PDE setting, the other relies on a particle closure ODE system resulting from a suitable regularization of the continuum description.
The second part of the paper focuses on a very different approach, which treats the original molecular system as an intrinsically mixed quantum-classical system.
Also in this case, we deal with both the continuum PDE setting and its particle ODE closure.

%%%%%%%%%%%%%%%%%%%%%%%%%%%%%%%%%%%%%%%%%%%%%%%%%%
%%%%%%%%%%%%%%%%%%%%%%%%%%%%%%%%%%%%%%%%%%%%%%%%%%
%%%%%%%%%%%%%%%%%%%%%%%%%%%%%%%%%%%%%%%%%%%%%%%%%%
\section{Exact factorization and variational asymptotics}\label{sec:Exact factorization and variational asymptotics}
Upon exploiting the XF of the molecular wavefunction, this section shows how BOMD can be derived by combining variational asymptotics and Euler-Poincar\'{e} reduction by symmetry, within the hydrodynamic formulation of the TDSE.
Our point of departure is the Dirac-Frenkel (DF) action principle $\delta\int_{t_1}^{t_2}\int_{\R^{3N}}\langle\Psi, i\hbar\partial_t\Psi-\widehat H_{mol}\Psi\rangle\,\mathrm{d}r\,{\mathrm{d}t=0}$ underlying the TDSE.
Here, we introduced the real-valued pairing $\langle\Psi_1,\Psi_2\rangle=\operatorname{Re}\langle\Psi_1|\Psi_2\rangle$, with $
	\langle\Psi_1|\Psi_2\rangle=\int_{\R^{3L}}\Psi_1(x,r)^*\Psi_2(x,r)\,\mathrm{d}x$.
A direct verification shows that the molecular TDSE is recovered from the DF action principle by using arbitrary variations.
As we are dealing with several nuclear masses, it is convenient to introduce the block-diagonal matrix $G:=\operatorname{diag}(M_1,\dots,M_N)\otimes\operatorname{Id}_{3\times 3}$, which induces a metric structure $g\colon\R^{3N}\times\R^{3N}\to\R$, operating as $g(u,v):=Gu\cdot v$.
In turn, this metric induces the norm $\|u\|_{g}:=\sqrt{g(u,u)}$ and its inverse $\|\sigma\|_{g^{-1}}:=\sqrt{g^{-1}(\sigma,\sigma)}$ on covectors $\sigma\in\R^{3N}$.
For example, the nuclear kinetic energy reads $\hbar^2\int_{\R^{3N}}g^{-1}(\nabla\Psi^\dagger,\nabla\Psi)\,\mathrm{d}r/2$, where the adjoint $\dagger$ is defined by the electronic inner product $\langle\cdot|\cdot\rangle$ above, so that $g^{-1}(\nabla\Psi^\dagger,\nabla\Psi)=\int_{\Bbb{R}^{3L}}\nabla\Psi^*(x)\cdot G^{-1}\nabla\Psi(x)\,\de x$.

\paragraph{Non-dimensionalization of the action principle.}
In order to obtain a hydrodynamic formulation of the molecular system, we proceed by applying the XF in \eqref{eq:XFA} and using the Madelung transform on the nuclear factor.
Before implementing these steps, however, it is convenient to perform a non-dimensionalization of the DF Lagrangian in order to reveal the role of the electron-to-nuclear mass ratio.
For this purpose, we introduce the unit system in \cite{Eich:2016} by replacing the electronic mass $m_e$ by the average nuclear mass $M_0=\sum_nM_n/N$.
Each physical unit is then expressed in terms of the mass $M_0$, the Bohr radius $\lambda_0$, the Hartree energy $E_h$ and the elementary charge $e$.
In particular, the units of time and action are given by $t_0=\sqrt{M_0/E_h}\lambda_0$ and $a_0=\sqrt{M_0E_h}\lambda_0$.
In addition, the reduced version of Planck constant $\hbar$ reads $\sqrt{\mu}a_0$, where we introduced $\mu:=m_e/M_0$.
Notice that we refrain from introducing an additional symbol to distinguish between original and dimensionless variables.
From now on all quantities are suitably non-dimensionalized.
At this point, we are left with the DF Lagrangian $L_{DF}=\int_{\R^{3N}}\langle\Psi,i\sqrt{\mu}\partial_t\Psi-\widehat H_{mol}\Psi\rangle\,\mathrm{d}r\,\mathrm{d}t$, where we have conveniently divided by the Hartree energy $E_h$.
The remainder of this paper will deal only with this non-dimensionalized version of the DF action principle and its variants.

\paragraph{Madelung transform and Euler-Poincar\'{e} variations.}
As a next step, we combine the XF \eqref{eq:XFA} of the dimensionless molecular wavefunction with the Madelung transform of the nuclear factor, that is $\Omega=\sqrt{D}e^{iS/\sqrt{\mu}}$.
These steps take the DF Lagrangian $\int_{\R^{3N}}\langle\Psi,i\hbar\partial_t\Psi-\widehat H_{mol}\Psi\rangle\,\mathrm{d}r$ into the form
\begin{multline}\label{eq:LXF}
	L_{XF}(S,D,\phi,\partial_t\phi)
	=\int\!\Big(\big\langle\phi,i\sqrt{\mu}\partial_t\phi-\widehat H_e\phi\big\rangle-\partial_t S\\
	-\frac{\mu}{2D}\|\nabla\sqrt{D}\|_{g^{-1}}^2-\frac{1}{2}\|\nabla S+\mathcal{A}_B\|_{g^{-1}}^2-\epsilon(\phi)\Big)D\mathrm{d}r,
\end{multline}
where $\mathcal{A}_B:=\left\langle\phi\mid-i\sqrt{\mu}\nabla\phi\right\rangle$ denotes the Berry connection and
\begin{align}\label{elecpot}
	\epsilon(\phi)
	:=\frac{\mu}{2}g^{-1}(\nabla\phi^\dagger,\nabla\phi)-\frac{1}{2}\|\mathcal{A}_B\|_{g^{-1}}^2
\end{align}
is usually referred to as the \emph{electronic potential}.
We notice that the latter can also be written as $\epsilon(\phi)={\mu}\operatorname{Tr}(G^{-1}\operatorname{Re}Q)/2$, where $Q_{jk}:=\langle\partial_j\phi|({1}-\phi\phi^\dagger)\partial_k\phi\rangle$ is the \emph{quantum geometric tensor}.
Also, in the reminder of this paper all integrals are on ${\R^{3N}}$, unless otherwise specified.

We observe that, so far, there seems to be nothing in the variational principle ensuring that the electronic factor is normalized at all times.
While this condition could be easily enforced by resorting to a Lagrange multiplier, this appears unnecessary at the current level.
Indeed, it follows from a direct verification that the normalization of $\phi$ is preserved in time by the equations resulting from $\delta\int_{t_1}^{t_2}L_{XF}\,{\rm d}t=0$.
Nevertheless, for later purpose, here we need to encode this normalization in the variational principle.
This is due to the fact that the time-conservation of the normalization condition may be lost when attempting to perform suitable approximations on the action, such as the asymptotic expansion that is presented below.
Instead of using Lagrange multipliers, here we follow the Euler-Poincar\'{e} reduction method in geometric mechanics \cite{Holm:1998}.
In particular, we restrict the electronic factor to evolve on orbits of the infinite-dimensional group ${\cal F}(\Bbb{R}^{3N},{\cal U}(L^2(\Bbb{R}^{3L})))$ of mappings $U\colon\Bbb{R}^{3N}\to{\cal U}(L^2(\Bbb{R}^{3L}))$ into the unitary operators on the electronic Hilbert space $L^2(\Bbb{R}^{3L})$.
This amounts to setting $\phi(t;r)=U_t(r)\phi_0(r)$ for some curve $U_t\in{\cal F}(\Bbb{R}^{3N},{\cal U}(L^2(\Bbb{R}^{3L})))$, so that the normalization of $\phi$ remains a preserved initial condition.
As customary in Euler-Poincar\'{e} theory, taking the relevant derivatives of $\phi$ yields $\partial_t\phi=\xi\phi$ and $\delta\phi=\gamma\phi$,	 where we have dropped the explicit time-dependence for convenience and both $\xi=\dot{U}U^{-1}$ and $\gamma=\delta{U}U^{-1}$ are skew-Hermitian operators on $L^2(\Bbb{R}^{3L})$.
With this in mind, the Lagrangian \eqref{eq:LXF} becomes
\begin{multline}\label{eq:LXF2}
	{\ell_{XF}}(S,D,\xi,\phi)
	=\int\!\Big(\big\langle\phi,i\sqrt{\mu}\xi\phi-\widehat H_e\phi\big\rangle-\partial_t S\\
	-\frac{\mu}{2D}\|\nabla\sqrt{D}\|_{g^{-1}}^2-\frac{1}{2}\|\nabla S+\mathcal{A}_B\|_{g^{-1}}^2-\epsilon(\phi)\Big)D\mathrm{d}r,
\end{multline}
where both $\delta S$ and $\delta D$ are arbitrary.
Having gone through several preparatory steps, we are ready to show how the BOMD equations can be derived from the action principle associated to \eqref{eq:LXF2}.

\paragraph{Variational asymptotics for BOMD.}
We will now prove that BOMD is recovered from the lowest-order asymptotic expansion of the Lagrangian \eqref{eq:LXF2}.
This expansion is obtained in the limit of a small electron-to-nuclear mass ratio, i.e., $\mu\to0$.
\begin{proposition}\label{fact:BOMD1}
	Consider the variational problem $\delta\int_{t_1}^{t_2}{\ell_{XF}}\,{\rm d}t=0$ associated to \eqref{eq:LXF2}, with $\delta\phi=\gamma\phi$ and arbitrary $\gamma$, $\delta D$ and $\delta S$.
	In the limit $\mu\to0$, this action principle yields the following continuum PDE system:
	\begin{align*}
		i)\,\frac{\partial D}{\partial t}+\operatorname{div}(DG^{-1}\nabla S)
		=0,\quad\,\,
		ii)\,
		\frac{\partial S}{\partial t}+\frac12\|\nabla S\|_{g^{-1}}^2
		=-E,\quad\,\,
		iii)\,\widehat H_e\phi
		=E\phi.
	\end{align*}
\end{proposition}
\begin{proof}
	To derive the above equations, we first consider the scaled Lagrangian arising in the limit $\mu\to0$, which, up to an overall sign, is given by
	\begin{align*}
		\tilde\ell_{XF}(S,D,\phi)
		=\int D\Big(\partial_t S+\langle\phi,\widehat H_e\phi\rangle+\frac{1}{2}\|\nabla S\|_{g^{-1}}^2\Big)\,\mathrm{d}r.
	\end{align*}
	The arbitrary variations in $\delta S$ yield the continuity equation i), while variations in $\delta D$ yield the nuclear Hamilton-Jacobi equation $\partial_ tS+\|\nabla S\|^2_{G^{-1}}/2=-\langle\phi,\widehat H_e\phi\rangle$.
	Although the Lagrangian $\tilde\ell_{XF}$ does not depend on $\xi$ because the term $\sqrt{\mu}D\langle\phi,i\xi\phi\rangle$ has vanished, we notice that the variations of the Euler-Poincar\'{e} variational principle $\delta\int_{t_1}^{t_2}\tilde\ell_{XF}\,{\rm d}t=0$ still contain the condition $\delta\phi=\gamma\phi$, which yields $[\phi\phi^\dagger,\widehat H_e]=0$.
	Applying both sides on $\phi$ results in the eigenvalue problem iii), thereby taking the Hamilton-Jacobi equation into the form ii).
\end{proof}

We remark that this conclusion can be equivalently reached by using arbitrary variations $\delta\phi$ and resorting to Lagrange multipliers.
The transport equation i) may be problematic due to the presence of second-order gradients of Hamilton's principal function $S$, which is known to develop caustic singularities in realistic scenarios.
At this point, one may introduce particle trajectories by proceeding formally and substituting $D(r,t)=\delta(r-R(t))$ into the transport equation i), thereby obtaining $G\dot R=\nabla S(R)$.
Then, taking the gradient of equation ii) and evaluating ${\rm d}\nabla S(R)/{\rm d} t=-\nabla E(R)$ leads to the Newtonian equations \eqref{Newton-eq}.
Alternatively, one can apply the standard method of characteristics, which provides solutions of the Hamilton-Jacobi equation the initial condition $S(r,t_0)=S_0(r)$.
In the present case, the characteristic trajectories of the Hamilton-Jacobi equation are given by the extremals of $\delta\int_{t_1}^{t_2}(\|\dot q\|^2_g/2-E(q))\,\mathrm{d}t$, which in particular satisfy the Newtonian equations \eqref{Newton-eq}.
Moreover, it is easy to prove that the momentum of the characteristics is given by the gradient $\nabla S$ of the solution $S$.

%%%%%%%%%%%%%%%%%%%%%%%%%%%%%%%%%%%%%%%%%%%%%%%%%%
%%%%%%%%%%%%%%%%%%%%%%%%%%%%%%%%%%%%%%%%%%%%%%%%%%
%%%%%%%%%%%%%%%%%%%%%%%%%%%%%%%%%%%%%%%%%%%%%%%%%%
\section{Quantum hydrodynamics and the bohmion method\label{sec:QHD+BM}}
In this section we show how the variational asymptotic method presented above may be applied in the context of quantum hydrodynamics and one of its particle closure schemes, recently proposed under the name of \emph{bohmion method} \cite{Foskett:2019}.
The latter hinges on a sampling process involving the Lagrangian paths underlying the relevant Madelung hydrodynamic equations.
These Lagrangian paths identify the well-known \emph{Bohmian trajectories} from quantum theory, thereby explaining the term `bohmion'.
The discussion will proceed in two stages.
First, we will present the hydrodynamic formulation of the XF equations as they arise from the action principe associated to \eqref{eq:LXF}.
In the second step, we will introduce the sampling process leading to the bohmion closure scheme and illustrate how the latter recovers BOMD by variational asymptotics.

\paragraph{Exact factorization and hydrodynamics.}
It is convenient to transform the Lagrangian \eqref{eq:LXF} by observing that the equation $\partial_t D+\operatorname{div}(DG^{-1}(\nabla S+\cA_B))=0$, obtained from the variations $\delta S$, identifies a Lie-transport evolution of the type $\de D/\de t=0$ along the vector field $u=G^{-1}(\nabla S+\cA_B)$.
Consequently, if the Lagrangian path $\eta_t\in\operatorname{Diff}(\Bbb{R}^{3N})$ is defined by the instantaneous integral curves of $u$ via the relation $\dot{\eta}_t(r_0)=u\circ\eta_t(r_0)$, then one can write $D(t)=\eta_*D_0$, where $\eta_*$ denotes the push-forward.
If the Lagrangian path is considered as a dynamical variable, then we can make the replacement $-\int_{\R^{3N}} D\partial_t S\de r=\int_{\R^{3N}} D\nabla S\cdot u\,\de r=:\int_{\R^{3N}} {M}\cdot u\,\de r$, in \eqref{eq:LXF2}.
Here, we have dropped an irrelevant total time derivative and we have introduced the hydrodynamic momentum ${M}=D\nabla S$.
Then, the action principle associated to \eqref{eq:LXF2} becomes
\begin{multline}\label{eq:LXF2-bis}
	\delta\int_{t_1}^{t_2}\!\int\!\Big({M}\cdot u+D\big\langle\phi,i\sqrt{\mu}\partial_t\phi-\widehat H_e\phi\big\rangle\\
	-\frac{\mu}{2}\|\nabla\sqrt{D}\|_{g^{-1}}^2-\frac{1}{2D}\|{M}+D\mathcal{A}_B\|_{g^{-1}}^2-D\epsilon(\phi)\Big)\mathrm{d}r=0,
\end{multline}
where $\delta{M}$ and $\delta\phi$ are both arbitrary.
Also, the variations 
\beq\label{vars}
	\delta D
	=-\operatorname{div}(Dw),\qquad\qquad
	\delta u
	=\partial_t w+u\cdot\nabla w-w\cdot\nabla u,
\eeq
are obtained directly from the relations $D(t)=\eta_*D_0$ and $u=\dot{\eta}_t\circ\eta_t^{-1}$ above, and by defining the arbitrary vector field $w=\delta\eta_t\circ\eta_t^{-1}$.

While the variational principle \eqref{eq:LXF2-bis} succeeds in unfolding the role of Lagrangian paths through the variations \eqref{vars}, we will perform two more steps that allow eliminating the explicit appearance of both the momentum variable ${M}$ and the Berry connection $\cA_B$ in the variational problem.
On the one hand, eliminating the momentum allows formulating the problem only in terms of the hydrodynamic velocity $u$.
On the other hand, eliminating $\cA_B$ is particularly convenient in obtaining a \emph{gauge-independent} formulation that removes the phase arbitrariness introduced by the XF decomposition \eqref{eq:XFA}.

\paragraph{Quantum motion in the hydrodynamic frame.}
We will proceed by expressing the quantum evolution of $\phi$ in the frame moving with the path $\eta_t$.
This process is performed as follows.
First, we observe that the quantum evolution reads \cite{Foskett:2019}
\begin{align}\label{phi-evol}
	i\hbar\left(\partial_t+u\cdot\nabla\right)\phi
	=\frac{1}{2D}\frac{\delta F}{\delta\phi}+\widehat H_e\phi,
	\quad\text{where}\quad
	F(D,\phi)
	:=\int D\epsilon(\phi)\,\mathrm{d}r.
\end{align}
As a further step, we define $\rho=\phi\phi^\dagger$ and write $\epsilon(\phi)=\mu\|\nabla\rho\|_{g^{-1}}^2/4=:\tilde\epsilon(\rho)$, where $\epsilon(\phi)$ is given in \eqref{elecpot} and $\|\nabla\rho\|_{g^{-1}}^2=\langle\nabla\rho,G^{-1}\nabla\rho\rangle$.
Also, we have denoted by $\langle A,B\rangle=\operatorname{Re}\langle A|B\rangle=\operatorname{Re}\operatorname{tr}(A^\dagger B)$ the pairing between trace-class operators on $L^2(\Bbb{R}^{3L})$.
The chain-rule relation ${\delta F}/{\delta\phi}=2({\delta\tilde F}/{\delta\rho})\phi$ takes \eqref{phi-evol} into the form
\beq\label{phi-evol2}
	i\hbar\left(\partial_t+u\cdot\nabla\right)\phi
	=\bigg(\frac{1}{D}\frac{\delta\tilde{F}}{\delta\rho}+\widehat H_e\bigg)\phi,
	\quad\text{where}\quad
	\tilde{F}(D,\rho)
	:=\int D\tilde\epsilon(\rho)\,\mathrm{d}r.
\eeq
Since the operator $D^{-1}{\delta\tilde{F}}/{\delta\rho}+\widehat H_e$ is Hermitian, the equation above unfolds the nature of the quantum evolution: a quantum state evolves unitarily in the frame moving with the hydrodynamic velocity $u$.
This conclusion motivates us to write the evolution of $\phi$ as $\phi(t)=(\tilde{U}_t\phi_0)\circ\eta_t^{-1}$, with $\tilde{U}_t\in{\cal F}(\Bbb{R}^{3N},{\cal U}(L^2(\Bbb{R}^{3L})))$.
Notice that the above relation identifies an action of the semidirect-product group $\operatorname{Diff}(\Bbb{R}^{3N})\,\circledS\,{\cal F}(\Bbb{R}^{3N},{\cal U}(L^2(\Bbb{R}^{3L})))$ that is constructed by the natural pullback action of $\operatorname{Diff}(\Bbb{R}^{3N})$ on ${\cal F}(\Bbb{R}^{3N},{\cal U}(L^2(\Bbb{R}^{3L})))$.
Consequently, the quantum evolution occurs on orbits of this semidirect-product group. 
Also, the evolution law $\phi=(\tilde{U}_t\phi_0)\circ\eta_t^{-1}$ implies $\partial_t\phi=\tilde\xi\phi-u\cdot\nabla\phi$, with $\tilde\xi=\partial_t{\tilde{U}}\tilde{U}^{-1}\circ\eta_t^{-1}$, and we observe that this is precisely the same form as in \eqref{phi-evol2}.

To proceed further, we replace $\langle\phi,i\partial_t\phi\rangle=\langle\phi,i\tilde\xi\phi\rangle+u\cdot\cA_B$ in \eqref{eq:LXF2-bis}.
Then, upon defining $m:={M}+D\cA_B$, and by inverting the Legendre transform $m=DGu$, the action principle \eqref{eq:LXF2-bis} becomes $\delta\int_{t_1}^{t_2}\ell_{EP}(u,D,\tilde\xi,\rho)\,\de t=0$ with the Euler-Poincar\'{e} Lagrangian
\beq\label{EPLagr1}
	\ell_{EP}
	=\int\!\Big(\frac{1}2D\|u\|_{g}^2+D\big\langle\rho,i\sqrt{\mu}\tilde\xi-\widehat H_e\big\rangle-\frac{\mu}{8D}{\|\nabla{D}\|_{g^{-1}}^2}-\frac\mu4D\|\nabla\rho\|_{g^{-1}}^2\Big)\,\mathrm{d}r.
\eeq
Here, the variations $\delta\rho=-\operatorname{div}(\rho w)+[\tilde\gamma,\rho]$ and $\delta\tilde\xi=\partial_t\tilde\gamma+[\tilde\gamma,\tilde\xi]-w\cdot\nabla\tilde\xi+u\cdot\nabla\tilde\gamma$ follow from the definitions $\rho=\phi\phi^\dagger$ and ${\tilde\xi=\partial_t{\tilde{U}}\tilde{U}^{-1}\circ\eta_t^{-1}}$, respectively.
Eventually, together with \eqref{vars}, the action principle associated to $\ell_{EP}$ leads to
\begin{align}\label{final-D-eqn}
	\begin{split}
		&(\partial_t+u\cdot\nabla)u^\flat
		=-\nabla V_Q-\langle\rho,\nabla\widehat{H}_e\rangle-\frac{\mu}{2D}\operatorname{div}\operatorname{tr}((D\nabla\rho)^{\sharp}\otimes\nabla\rho),\\
		&i\hbar(\partial_t+u\cdot\nabla)\rho
		=\Big[\widehat{H}_e-\frac{\mu}{2D}{\rm div}(D\nabla\rho)^\sharp,\rho\Big],\qquad\,
		\partial_t D+{\rm div}(Du)
		=0.
	\end{split}
\end{align}
Here, we used the musical isomorphisms induced by the metric $g$ while $\operatorname{tr}$ is the quantum trace.
Also, ${V_Q=-\mu\operatorname{div}(\nabla\sqrt{D})^\sharp/(2\sqrt{D})}$ is the \emph{quantum potential}.

The system \eqref{final-D-eqn} represents the hydrodynamic form of the equations of motion resulting from the variational principle for \eqref{eq:LXF}.
The second-order gradients in the first two equations make the level of complexity of this system rather intimidating.
A finite-dimensional closure scheme becomes necessary in order for these equations to be used in the context of molecular dynamics simulations.

\paragraph{Regularization and the bohmion scheme.}
An immediate consequence of the second-order gradients in \eqref{final-D-eqn} is the lack of delta-type solutions of the form $D(r,t)=\delta(r-R(t))$, which are instead allowed after taking the classical limit $\mu\to0$ as shown in Proposition~\ref{fact:BOMD1}.
To overcome this limitation, the bohmion method exploits a variational regularization to restore point-particle trajectories.
This regularization is applied after rewriting the Euler-Poincar\'{e} Lagrangian \eqref{EPLagr1} in terms of the weighted variable $\tilde\rho=D\rho$.
Then, one smoothens by replacing ${\|\nabla{D}\|_{g^{-1}}^2}/{D}\to{\|\nabla{\bar{D}}\|_{g^{-1}}^2}/{\bar{D}}$ and ${\|\nabla\tilde\rho\|_{g^{-1}}^2}/{D}\to\|\nabla\bar{\rho}\|_{g^{-1}}^2/{\bar{D}}$, where $\bar{D}=K_\alpha*D$ and $\bar{\rho}=K_\alpha*\tilde\rho$, for some normalized convolution kernel $K_\alpha$ depending on a modeling length-scale $\alpha>0$.
In the limit $\alpha\to0$, we ask for $K_\alpha$ to tend to a delta function thereby recovering the original Lagrangian.

Due to this smoothing process, the resulting regularized equations allow for singular delta-like expressions of $D$ and $\tilde\rho$, thereby returning point trajectories called \emph{bohmions}.
The trajectory equations may be found by replacing the ansatz $D(r,t)=\sum_{a=1}^P w_a\delta(r-{q}_a(t))$ and $\tilde\rho(r,t)=\sum_{a=1}^P w_a\varrho_{a}(t)\delta(r-{q}_a(t))$ in the regularized Lagrangian, where the positive weights $w_a$ satisfy $\sum_{a}w_a=1$.
Also, we have $\varrho_a=\varphi_a\varphi_a^\dagger$ and $\varrho_a(t)=U_a(t)\varrho_{0a}U_a(t)^\dagger$, so that $\partial_t\varrho_a=[\xi_a,\varrho_a]$.
Upon denoting $\widehat H_a=\widehat H_e({{q}}_a)$, this process leads to the Lagrangian
\begin{multline*}
	L(q,\dot{q},\rho)
	=\sum_{a=1}^Pw_a\bigg(\frac{1}{2}\|\dot{{q}}_a\|_{g}^2+\langle\varrho_a, i\sqrt{\mu}\xi_a-\widehat H_a\rangle\\
	+\frac{\mu}{8}\sum_{b=1}^P w_b(1-2\langle\varrho_a,\varrho_b\rangle)\int\frac{g^{-1\!}\big(\nabla K_\alpha(r-{{q}}_a),\nabla K_\alpha(r-{{q}}_b)\big)}{\sum_c w_cK_\alpha(r-{{q}}_c)}\,\mathrm{d}r\bigg),
\end{multline*}
where $\delta\xi_a=\partial_t\gamma_a-[\xi_a,\gamma_a],\,\delta\varrho_a=[\gamma_a,\varrho_a]$, while $\delta q_a$ and $\gamma_a$ are arbitrary.
We emphasize that these bohmions do not correspond to physical particles, but rather to \emph{computational particles} that sample nuclear hydrodynamic paths.

It is now easy to see how the bohmion scheme reduces to BOMD in the classical limit.
Indeed, letting $\mu\to0$ and taking variations of the resulting Lagrangian ${\tilde L}=\sum_{a}w_a\big(\|\dot{{q}}_a\|_{g}^2/2-\langle\varrho_a,\widehat H_a\rangle\big)$ yields
\begin{align}\label{eq:bohmions}
	G\ddot{{q}}_a
	=-\nabla\langle\varphi_a,\widehat H_e({{q}}_a)\varphi_a\rangle,\qquad\quad
	\widehat H_e({{q}}_a)\varphi_a
	=E({{q}}_a)\varphi_a.
\end{align}
Here, the second equation follows from the Euler-Poincar\'{e} equation $[\widehat H_a,\varrho_a]=0$, which arises from the variations $\delta\xi_a$ in Hamilton's principle.
We observe that BOMD is recovered in the case of only one bohmion, that is $P=1$.
Importantly, we have obtained an on-the-fly implementation scheme for BOMD in which the electronic structure problem associated to \eqref{eq:EEVP} (usually very challenging) is replaced by a finite-dimensional eigenvalue problem to be solved at each time-step along trajectories.
In the general case ${P>1}$, the bohmion method recovers an on-the-fly model to BOMD in which the single particle distribution ${D(r,t)=\delta(r-R(t))}$ is replaced by the statistical sampling $D(r,t)=\sum_{a=1}^P w_a\delta(r-{q}_a(t))$ of the entire density associated to classical nuclear motion.

%%%%%%%%%%%%%%%%%%%%%%%%%%%%%%%%%%%%%%%%%%%%%%%%%%
%%%%%%%%%%%%%%%%%%%%%%%%%%%%%%%%%%%%%%%%%%%%%%%%%%
%%%%%%%%%%%%%%%%%%%%%%%%%%%%%%%%%%%%%%%%%%%%%%%%%%
\section{Mixed quantum-classical dynamics}
This section extends the variational asymptotic methods discussed earlier to the case of \emph{mixed quantum-classical} (MQC) models.
This type of models are motivated by the need to go beyond BOMD when the BO approximation fails to hold.
Inspired by the result from BO theory that nuclear dynamics can be approximated as classical one seeks a model in which classical nuclear motion is coupled to fully quantum electronic evolution.
Here, we focus on the model presented in \cite{Gay-Balmaz:2022aa}, to which we refer for a thorough discussion.
The main ingredient of this model resides in {\it Koopman wavefunctions}, that is functions $\chi(r,p)\in L^2(\Bbb{R}^{6N})$ such that $\rho_c=|\chi|^2$ obeys the classical Liouville equation $\partial_t\rho_c=\{H,\rho_c\}$.
Then, one can take the tensor-product space to describe MQC dynamics in terms of hybrid wavefunctions $\Upsilon\in L^2(\Bbb{R}^{6N})\otimes L^2(\Bbb{R}^{3L})$ so that $\rho_c=\Upsilon^\dagger\Upsilon$ and $\hat\uprho_q=\int\Upsilon\Upsilon^\dagger\,\de r\de p$ are the classical Liouville density and the quantum density matrix, respectively.

Instead of presenting the model equations, here we follow the procedure outlined before and consider their underlying variational principle \cite{Gay-Balmaz:2022aa}.
For the model under consideration, the latter involves the following Lagrangian $L_{\operatorname{QC}}(\Upsilon,\partial_t\Upsilon,\mathcal{X})$:
\begin{align*}
	L_{\operatorname{QC}}
	=\int\left\langle\Upsilon,i\hbar\partial_t\Upsilon+(\mathcal{A}-\mathcal{A}_B)\cdot\mathcal{X}\Upsilon-({\widehat{H}}-\mathcal{A}_B\cdot{{X}}_{\widehat{H}})\Upsilon+i\hbar{{X}}_{\widehat{H}}\cdot\nabla\Upsilon\right\rangle\,\de r\de p,
\end{align*}
where we have introduced $\mathcal{A}_B=\left\langle\Upsilon,-i\hbar\nabla\Upsilon\right\rangle/(\Upsilon ^\dagger\Upsilon)$ and the MQC Hamiltonian vector field ${X}_{\widehat H}:=(\partial_p\widehat H,-\partial_r\widehat H)$ of the operator-valued Hamiltonian function $\widehat H(r,p)$.
Also, $\mathcal{A}=(p,0)$ is the coordinate representation of the canonical Liouville one-form $\mathcal{A}=p\cdot\de r$.
In the above Lagrangian, the variation $\delta\Upsilon$ is arbitrary, while the variation of the vector field $\mathcal{X}$ arises from its definition in terms of the Lagrangian phase-space paths $\upeta_t\in\operatorname{Diff}(\Bbb{R}^{6N})$, that is $\dot\upeta_t(r_0,p_0)=\mathcal{X}\circ\dot\upeta_t(r_0,p_0)$.
Then, we have the Euler-Poincar\'{e} variation
\begin{align}\label{eq:varX}
	\delta\mathcal{X}
	=\partial_t\mathcal{Y}+\mathcal{X}\cdot\nabla\mathcal{Y}-\mathcal{Y}\cdot\nabla\mathcal{X},
\end{align}
where $\mathcal{Y}$ is an arbitrary displacement vector field.

Performing the non-dimensionalization of the MQC Lagrangian, applying the XF $\Upsilon(r,p,t)=\Omega(r,p,t)\phi(x,t;r,p)$ as in \eqref{eq:XFA}, and using the Madelung transform $\Omega=\sqrt{\rho_c}e^{iS/\sqrt{\mu}}$, we rewrite $L_{\operatorname{QC}}$ into the form
\begin{align}\label{eq:hybridnon}
	\ell_{\operatorname{QC}}\!
	=\int\!\bigg(\rho_c\Big(\dot S+(\nabla S-\mathcal{A})\cdot\mathcal{X}+\langle\phi,\widehat H_e\phi\rangle+\frac{1}{2}\|p\|_{g^{-1}}^2\Big)+\mathcal{O}(\sqrt{\mu})\bigg)\de r\de p,
\end{align}
where we decomposed the MQC Hamiltonian as $\widehat H=\|p\|_{g^{-1}}^2/2+\widehat H_e$ into the classical kinetic energy operator for the nuclei and the quantum electronic part.
Similarly to the discussion preceding Proposition~\ref{fact:BOMD1}, here we have retained the normalization condition for $\phi$ by letting $\phi=U\phi_0$.
The latter implies $\delta\phi=\gamma\phi$, where $\gamma$ is an arbitrary skew-Hermitian operator.
A direct verification leads to
\begin{proposition}
	Consider the variational problem $\delta\int_{t_1}^{t_2}\ell_{\operatorname{QC}}\,{\rm d}t=0$ associated to \eqref{eq:hybridnon}, with $\delta\phi=\gamma\phi$, $\delta\mathcal{X}$ as in \eqref{eq:varX}, and arbitrary $\delta D$ and $\delta S$.
	In the limit $\mu\to0$, this action principle yields the following continuum PDE system:
	\begin{align*}
		i)\,\partial_t\rho_c+\operatorname{div}(\rho_c{X}_{ H_{cl}})
		=0,\quad
		ii)\,\partial_t S+\{S, H_{cl}\}
		=p\partial_p H_{cl}-H_{cl},\quad
		iii)\,\widehat H_e\phi
		=E\phi,
	\end{align*}
	where $ H_{cl}({z})=\frac{1}{2}\|p\|_{g^{-1}}^2+E(q)$.
\end{proposition}
We observe that the phase $S$ decouples completely and the classical Liouville equation is solved by $\rho_c(r,p,t)=\delta(r-R(t))\delta(p-P(t))$, thereby recovering the phase-space formulation of the BOMD equations \eqref{Newton-eq}.

Similarly to the construction of the bohmion code, one can formulate a trajectory-based closure for the MQC model.
In particular, we refer to \cite{Tronci:2023} for details on the \emph{koopmon method}.
In short, a regularization process similar to the one in Section~\ref{sec:QHD+BM} leads to the koopmon Lagrangian 
\begin{align*}
	L(r,p,\rho)
	=\sum_{a}w_a\bigg(p_a\dot q_a+\bigg\langle\rho_a,i\sqrt{\mu}\xi_a-\widehat H(q_a,p_q)-\frac{i\sqrt{\mu}}2\sum_bw_b\left[\rho_b,\mathcal{I}_{ab}\right]\bigg\rangle\bigg),
\end{align*}
where $\mathcal{I}_{ab}:=\int{K_a\{K_b,\widehat H\}}/({\sum_cw_cK_c})\,\de r\de p$ and $K_s(r,p):=K(r-q_s,p-p_s)$.
Performing the limit $\mu\to0$ gives us $\tilde L=\sum_aw_a(p_a\dot q_a-\langle\rho_a,\widehat H(q_a,p_a)\rangle)$, which returns the phase-space form of the bohmion implementation \eqref{eq:bohmions} of BOMD.

% ---- Bibliography ----

\newpage
\begin{center}
	\itshape\bfseries\Large -- Supplementary Material --
\end{center}

\appendix
%%%%%%%%%%%%%%%%%%%%%%%%%%%%%%%%%%%%%%%%%%%%%%%%%%
%%%%%%%%%%%%%%%%%%%%%%%%%%%%%%%%%%%%%%%%%%%%%%%%%%
%%%%%%%%%%%%%%%%%%%%%%%%%%%%%%%%%%%%%%%%%%%%%%%%%%
\section{Non-dimensionalization of the XF Lagrangian}
To show how the equations of BOMD are recovered in the classical limit ${\mu\to 0}$, we combined the non-dimensionalization of the XF Lagrangian with the Madelung transform of the dimensionless nuclear factor.
The following describes in more detail how these steps lead to the Lagrangian
\begin{multline}\label{eq:LXF1}
	L_{XF}(S,D,\phi,\partial_t\phi)
	=\int_{\R^{3N}}\!\Big(\big\langle\phi,i\sqrt{\mu}\partial_t\phi-\widehat H_e\phi\big\rangle-\partial_t S\\
	-\frac{\mu}{2D}\|\nabla\sqrt{D}\|_{g^{-1}}^2-\frac{1}{2}\|\nabla S+\mathcal{A}_B\|_{g^{-1}}^2-\epsilon(\phi)\Big)D\mathrm{d}r,
\end{multline}
which was presented in Section~\ref{sec:Exact factorization and variational asymptotics}.
First, we note that the substitution of the exact factorization ansatz $\Psi(r,x,t)
	=\Omega(r,t)\phi(x,t;r)$ for the molecular wave function into the Dirac-Frenkel Lagrangian $\int_{\R^{3N}}\langle\Psi,i\hbar\partial_t\Psi-\widehat H_{mol}\Psi\rangle\,\mathrm{d}r$ yields
\begin{align*}
	L_{XF}
	&=\operatorname{Re}\int_{\R^{3N}}\Bigg(i\hbar\Omega^*\partial_t\Omega+|\Omega|^2\left\langle\phi\mid i\hbar\partial_t\phi-\hat H_e\phi\right\rangle\\
	&\qquad-\frac{\hbar^2}{2}\Big[\|\nabla\Omega\|_{G^{-1}}^2+|\Omega|^2\|\nabla\phi\|_{g^{-1}}^2-2\Omega^*\left\langle\phi\mid g^{-1}\left(\nabla\Omega,\nabla\phi\right)\right\rangle\Big]\Bigg)\,\mathrm{d}r.
\end{align*}
Next, in order to remove all physical dimensions in this Lagrangian, we express physical variables in our new unit system.
For example, the physical dimension of the nuclear wave function $\Omega(r,t)$ is removed by introducing a new wave function $\Omega'(r',t')$ that depends only on the dimensionless variables $r'$ and $t'$.
Note that we introduced the prime symbol to indicate that a given variable is dimensionless.
The relation between $\Omega'$ and $\Omega$ is then $\Omega'(r',t')=\Omega(r'\lambda_0,t't_0)\lambda_0^{3N/2}$, where we used that the unit of $\Omega$ is $\lambda_0^{-3N/2}$.
From this relation it follows that the first term of the above Lagrangian can be written, as follows:
\begin{align*}
	i\hbar\Omega^*\partial_t\Omega
	=i\sqrt{\mu}a_0\Omega'^*\lambda_0^{-3N/2}\partial_{t'}\Omega'\lambda_0^{-3N/2}t_0^{-1}
	=i\sqrt{\mu}\Omega'^*\partial_{t'}\Omega'\lambda_0^{-3N}E_h.
\end{align*}
Similarly, we rewrite all the other terms.
Performing the change of variables $r=r'\lambda_0$ in the integral, a simple calculation shows that the non-dimensionalized form of the Lagrangian is given by
\begin{align*}
	L_{XF}'
	&=\operatorname{Re}\int_{\R^{3N}}\Bigg(i\sqrt{\mu}\Omega'^*\partial_{t'}\Omega'+|\Omega'|^2\left\langle\phi'\mid i\sqrt{\mu}\partial_{t'}\phi'-\hat H_e'\phi'\right\rangle\\
	&-\frac{\mu}{2}\Big[\|\nabla'\Omega'\|_{g'^{-1}}^2+|\Omega'|^2\|\nabla'\phi'\|_{g'^{-1}}^2-2\Omega'^*\left\langle\phi'\mid (g^{-1})'\left(\nabla'\Omega',\nabla'\phi'\right)\right\rangle\Big]\Bigg)\,\mathrm{d}r',
\end{align*}
where $L_{XF}'=L_{XF}/E_h$.
Since now all quantities are suitably non-dimensionalized, it is convenient to omit the prime symbol in the further calculations.
Applying the Madelung transform $\Omega=\sqrt{D}e^{iS/\sqrt{\mu}}$ for the dimensionless nuclear wave function and taking advantage of the fact that all purely imaginary terms can be removed by the real part operation, one finally arrives at
\begin{align*}
	L_{XF}
	&=-\int_{\R^{3N}}\Bigg(D\partial_t S-D\left\langle\phi,i\sqrt{\mu}\partial_t\phi-\hat H_e\phi\right\rangle\\
	&\qquad+\frac{\mu}{2}\|\nabla\sqrt{D}\|^2_{g^{-1}}+\frac{D}{2}\|\nabla S\|^2_{g^{-1}}+\frac{D\mu}{2}\|\nabla\phi\|_{g^{-1}}^2+Dg^{-1}(\nabla S,\mathcal{A}_B)\Bigg)\,\mathrm{d}r\\
	&=\int_{\R^{3N}}\!\Big(\big\langle\phi,i\sqrt{\mu}\partial_t\phi-\widehat H_e\phi\big\rangle-\partial_t S\\
	&\qquad-\frac{\mu}{2D}\|\nabla\sqrt{D}\|_{g^{-1}}^2-\frac{1}{2}\|\nabla S+\mathcal{A}_B\|_{g^{-1}}^2-\epsilon(\phi)\Big)D\mathrm{d}r,
\end{align*}
which is exactly the Lagrangian in \eqref{eq:LXF1}.

%%%%%%%%%%%%%%%%%%%%%%%%%%%%%%%%%%%%%%%%%%%%%%%%%%
%%%%%%%%%%%%%%%%%%%%%%%%%%%%%%%%%%%%%%%%%%%%%%%%%%
%%%%%%%%%%%%%%%%%%%%%%%%%%%%%%%%%%%%%%%%%%%%%%%%%%
\section{Proof of proposition 2}
In the first step, we derive the electronic equation by following the same arguments as presented in the proof of Proposition~\ref{fact:BOMD1}:
Using that the $\mathcal{O}(\sqrt{\mu})$ term in $\ell_{QC}$ contains the summand $\sqrt{\mu}D\langle\phi,i\dot\phi\rangle$, we express the time derivative $\dot\phi$ in terms of $\xi:=\dot UU^{-1}$ as $\dot\phi=\xi\phi$.
Since in the limit $\mu\to 0$ the variations of the Euler-Poincar\'{e} variational principle still contain the condition $\delta\phi=\eta\phi$, where $\eta=(\delta U)U^{-1}$ is arbitrary, we then obtain the commutator relation $[\phi\phi^\dagger,\widehat{H}_e]=0$.
Note that here $\phi$ depends parametrically on ${z}=(q,p)$, while $\widehat H_e$ depends only on the nuclear position $q$.
Since the commutator relation must hold for all ${z}$, it therefore follows that the solutions can be written independently of the nuclear momentum.
More precisely, each solution $\phi(x;z)$ can be identified as a function $\tilde\phi\colon\R^{3N}\to L^2(\R^{3L})$ that depends only on $q$, that is, $\phi(x;z)=\tilde\phi(x;q)$ for all $p\in\R^{3N}$.
Hence, upon applying both sides of $[\phi\phi^\dagger,\widehat{H}_e]=0$ to $\phi$, results in the familiar eigenvalue problem iii).
In the second step, we derive the nuclear equations.
Variations in $\delta S$ yield the continuity equation $\partial_tD+\operatorname{div}(D\mathcal{X})=0$, while variations in $\delta D$ yield
\begin{align}\label{eq:dotS1}
	\partial_tS+(\nabla S-\mathcal{A})\cdot\mathcal{X}+\frac{1}{2}\|p\|_{g^{-1}}^2+E(q)
	=0.
\end{align}
Using that fact that $\delta\mathcal{X}=\partial_t\mathcal{Y}+\mathcal{X}\cdot\nabla\mathcal{Y}-\mathcal{Y}\cdot\nabla\mathcal{X}$, we conclude that $\partial_t(\nabla S-\mathcal{A})+\mathcal{X}\cdot\nabla(\nabla S-\mathcal{A})+\nabla\mathcal{X}\cdot(\nabla S-\mathcal{A})
	=0$.
In the next step we determine the unknown vector field $\mathcal{X}$.
Using $\partial_t(\nabla S)=\nabla(\partial_t S)$ and inserting the expression for $\partial_tS$ into \eqref{eq:dotS1}, a short calculation shows that
\begin{align*}
	\mathcal{X}\cdot\nabla\mathcal{A}+\nabla\mathcal{X}\cdot\mathcal{A}-\nabla(\mathcal{A}\cdot\mathcal{X})
	=-\left(\nabla_q E(q),G^{-1}p\right).
\end{align*}
Since $\mathcal{X}\cdot\nabla\mathcal{A}+\nabla\mathcal{X}\cdot\mathcal{A}-\nabla(\mathcal{A}\cdot\mathcal{X})=\left((\nabla\mathcal{A})^T-\nabla\mathcal{A}\right)\mathcal{X}=J\mathcal{X}$, where $J=(\nabla\mathcal{A})^T-\nabla\mathcal{A}$ is the canonical symplectic form, we conclude that $\mathcal{X}={X}_{ H_{cl}}$.
Consequently, inserting this expression for the transport vector field $\mathcal{X}$ into \eqref{eq:dotS1}, we obtain the Koopman-van Hove (KvH) equation ii) for the phase $S$.\\[-5mm]

$\hfill\blacksquare$

\bigskip\noindent
Combining the transport equation i) with the phase equation ii) for the phase $S$, we obtain the following \emph{Koopman-van Hove} equation for the nuclear dynamics:
\begin{align*}
	i\sqrt{\mu}\Omega
	=i\sqrt{\mu}\{H_{cl},\Omega\}-(p\partial_p H_{cl}-H_{cl}).
\end{align*}

\end{document}